\documentclass[11pt]{article}
\usepackage{amsmath,cite}
\usepackage{mathrsfs}
\usepackage{longtable}
\usepackage{color}
\newcommand{\Z}{{\bf Z}}

\newcommand{\T}{{\cal T}}

\usepackage{longtable}
\usepackage{graphicx}
\usepackage{amsmath}
\usepackage{amsfonts}
\usepackage{amssymb}%
\usepackage{multicol}
\usepackage{cuted}

\newtheorem{theorem}{Theorem}[section]
\newtheorem{definition}[theorem]{Definition}
\newtheorem{lemma}[theorem]{Lemma}

\newtheorem{corollary}[theorem]{Corollary}

\newenvironment{proof}[1][Proof]{\textbf{#1.} }{\ \rule{0.5em}{0.5em}}

\def\whitebox{{\hbox{\hskip 1pt
    \vrule height 6pt depth 1.5pt
    \lower 1.5pt\vbox to 7.5pt{\hrule width
            3.2pt\vfill\hrule width 3.2pt}%
    \vrule height 6pt depth 1.5pt
    \hskip 1pt } }}
\def\qed{\ifhmode\allowbreak\else\nobreak\fi\hfill\quad\nobreak
         \whitebox\medbreak}

\begin{document}

\baselineskip 17pt

\title{Constructions and Applications of Perfect Difference Matrices and Perfect Difference Families}

\author{Xianwei Sun$^1,$~~~~Huangsheng Yu$^2,$~~~~Dianhua Wu$^{2}$~~~~}

\date{}
\maketitle
\noindent {\small $1$ Hubei Key Lab of Transportation Internet of Things, Wuhan University of Technology,
Wuhan 430070, China;}\\
\noindent {\small $2$ School of Mathematics and Statistics, Guangxi Normal University,
Guilin 541006, China.}\\


\date{}
\maketitle \noindent {\bf Abstract}
Perfect difference families (PDFs for short) are important both in theoretical and in applications.
Perfect difference matrices (PDMs for short) and the equivalent structure had been extensively studied and used to
construct perfect difference families, radar array and related codes.
The necessary condition for the existence of a PDM$(n,m)$ is $m\equiv 1\pmod2$ and $m\geq n+1$. So far,
PDM$(3,m)$s  exist for odd $5\leq m\leq 201$ with two definite exceptions of $m=9,11$.
In this paper, new recursive constructions on PDM$(3,m)$s are
investigated, and it is proved that there exist PDM$(3,m)$s  for any odd $5\leq m<1000$ with two definite exceptions of $m=9,11$
and $33$ possible exceptions.
A complete result of $(g,\{3,4\},1)$-PDFs with the ratio of block size  $4$ no less than
$\frac{1}{14}$ is obtained. As an application, a complete class of perfect  strict optical
orthogonal codes with weights $3$ and $4$ is obtained.


\section{Introduction}

Let ${\mathbf Z}_v$ be the residue ring of integers modulo $v$.
The {\em directed list of differences} of a subset $C=\{c_1, c_2, \ldots,c_t\}$ of ${\mathbf Z}_v$ will be denoted by
the multiset $\Delta C=\{c_{i} - c_{j} \ | \  1 \le j< i \le t\}$.
More generally, the directed list of differences of a set ${\cal C}$ of subsets of ${\mathbf Z}_v$ is the multiset
$\Delta {\cal C}=\bigcup\limits_{C\in {\cal C}}\Delta C$.

In this paper, we will always assume that $v$ is odd, and we will also use $\Delta {\cal C}$
to denote the directed list of differences of a set ${\cal C}$ of subsets of ${\mathbf Z}_v$.

Let $K = \{k_1, k_2, \ldots, k_s\}$ be a set of  positive
integers,  ${\cal B} = \{B_1,B_2,\ldots,B_h\}$ be a collection of subsets of ${\mathbf Z}_v$ called {\em blocks},
 and $L \subseteq \{1,2,\ldots, \frac{v-1}{2}\}$.
If the list of directed differences ${\Delta}{\cal B}$
covers each element of the set $\{1,2,\ldots, \frac{v-1}{2}\} \setminus L$ exactly once, then we call ${\cal B}$
a $(v,K,1)$ {\em perfect difference packing}, or $(v,K,1)$-PDP, with {\em difference leave} $L$.
A $(v,K,1)$-PDP, ${\cal B} = \{B_1,B_2,\ldots,B_h\}$, with difference leave $L=\emptyset$,
is called a $(v,K,1)$ {\em perfect difference family}, or briefly a $(v,K,1)$-PDF. When $K = \{k\}$, the notation $(v,k,1)$-PDF is used.


Perfect difference families are in fact a special case of perfect systems of difference sets.
Perfect systems of difference sets were first introduced
in \cite{BKT76, KT77}
in connection with a problem of spacing movable antennas in radioastronomy.
Let $c, m, p_1, \ldots, p_m$ be positive integers, and  ${\cal S} =\{S_1, S_2, \ldots, S_m\}$,
where $S_i = \{s_{i1}, s_{i2}, \ldots, s_{ip_i}\}$, $0 \le s_{i1} < s_{i2} < \cdots < s_{ip_i}$,
and all $s_{ij}$'s are integers.
We say that  ${\Delta}{\cal S} = \{{\Delta}S_1, {\Delta}S_2, \ldots, {\Delta}S_m\}$ is a
{\em perfect system of difference sets for $c$} (or {\em starting with $c$}, or {\em with threshold $c$}),
or briefly, an $(m,\{p_1,p_2,\ldots,p_m\},c)$-PSDS, if
${\Delta}{\cal S} = \{c,c+1,\ldots, c-1+\sum\limits_{1 \le i \le m}{p_i \choose 2}\}$.
Each subset $S_i$ is called a {\em block} and each set ${\Delta}S_i$ a {\em component} of the system.
An $(m,\{p_1,p_2,\ldots,p_m\},c)$-PSDS  is {\em regular} if $p_1 = p_2 = \cdots = p_m = p$.
As usual, a regular $(m,\{p\},c)$-PSDS is abbreviated to $(m,p,c)$-PSDS.
Obviously, a $(v,k,1)$-PDF is a regular $(\frac{v-1}{k(k-1)},k,1)$-PSDS.


In spite of extensive efforts put into perfect systems of difference sets
(see, for examples \cite{AB04, Ab82, Ab84, KT77, KT79, Ma87, Ro92, Tu80, Tu82}),
known constructions and existence results on this topic are not rich.
The necessary condition for the existence of a $(\frac{v-1}{k(k-1)},k,1)$-PSDS,
i.e. a $(v,k,1)$-PDF, is $v \equiv 1 \pmod{k(k-1)}$.
In \cite{BKT76, KT80}, it is  proved that
perfect difference families cannot exist for $k \ge 6$.
For $k = 3$, the existence problem has been completely settled:
a $(v,3,1)$-PDF exists if and only if $v \equiv 1,7 \pmod{24}$.
For $k=4$, however, the existence problem is far from settled: we only know that there exists a $(12t+1,4,1)$-PDF  for $t=1,4$-$1000$
and there exist no $(12t+1,4,1)$-PDF for $t=2,3$, see \cite{GMS2010} and the references therein.
For $k=5$, the existence results are even scarce.
The interested readers are referred to \cite[p. 400]{AB07} for a recent survey
of  perfect difference families.

It is noted that perfect difference families are also
closely related to many other concepts such as cyclic difference families, difference triangle sets,
optical orthogonal codes, and strict optical orthogonal codes.
A {\em cyclic} $(v,k,1)$  difference family \cite {AB07}, briefly a {\em cyclic} $(v,k,1)$-DF,
is a collection ${\cal F}$ of $k$-subsets of ${\mathbf Z}_v$ such that
${\Delta}{\cal F} \cup (-{\Delta}{\cal F}) =  {\mathbf Z}_v \setminus \{0\}$.
A $(v,k,1)$-PDF immediately implies a cyclic $(v,k,1)$-DF.
An $(n,k)$ {\em difference triangle set} \cite{Sh07},
or $(n,k)$-DTS (the notion D${\Delta}$S is used in \cite{Sh07}),
is a collection ${\cal S}$ of $n$ $(k+1)$-subsets of ${\mathbf Z}_v$
such that the elements in ${\Delta}{\cal S}$ are all distinct positive integers.
The {\em scope} of ${\cal S}$ is the maximum of ${\Delta}{\cal S}$,
so the scope is at least equal to $n{k+1 \choose 2}$.
In the case that this lower bound is met,
${\cal S}$ can be viewed as a $(2n{k+1 \choose 2}+1,k+1,1)$-PDF.
A $(v,k,1)$ {\em optical orthogonal code} \cite{He07}, or $(v,k,1)$-OOC, is a collection ${\cal C}$
of $k$-subsets of ${\mathbf Z}_v$ such that ${\Delta}{\cal F} \cup (-{\Delta}{\cal F})$
does not have repeated elements in ${\mathbf Z}_v \setminus \{0\}$
and the set-wise stabilizer of each $k$-subset of ${\cal C}$ is the subgroup $\{0\}$ of ${\mathbf Z}_v$.
A $(v,k,1)$-OOC is {\em optimal} if its size reaches the upper bound $\lfloor\frac{v-1}{k(k-1)}\rfloor$.
It can be easily seen that a $(v,k,1)$-PDF gives  optimal $(u,k,1)$-OOCs
with  $v\leq u\leq v+k(k-1)-1$.

Kotzig and Turgeon \cite{KT79} discovered that arbitrarily large perfect systems of difference sets
can be constructed from smaller ones via additive sequences of permutations.
Let $X^{(1)} = (x_1^{(1)}, \ldots, x_m^{(1)})$ be an ordered set of distinct integers.
For $j = 2, \ldots, n$, let $X^{(j)} = (x_1^{(j)}, \ldots, x_m^{(j)})$
be a permutation of distinct integers in $X^{(1)}$.
Then the ordered set $(X^{(1)}, X^{(2)}, \ldots, X^{(n)})$ is called an
{\em additive sequence of permutations} of {\em length} $n$ and {\em order} $m$, ASP$(n,m)$ for short,
if for every subsequence of consecutive permutations of the ordered set $X^{(1)}$,
their vector-sum is again a permutation of $X^{(1)}$.
The set  $X^{(1)}$ is usually called the {\em basis} of the additive sequence of permutations. In this paper, we will always consider
ASP$(n,m)$ with base $X^{(1)}=I_m=\{-r,-(r-1),...,-1,0,1,...,r-1,r\}$, where $m=2r+1$ and $r$ a positive integer, unless otherwise stated.
It is noted (see, for examples \cite{Ab84, Tu80, Tu82}), conversely, that certain perfect systems of difference sets can also be used to construct additive sequences of permutations.

In \cite{GMS2010,WC2010}, a construction of  additive sequence of permutations via perfect difference matrix was introduced. An $n \times m$ matrix $D=(d_{ij})$ with entries from $I_m$, is called a {\em perfect difference matrix}, denoted by PDM$(n,m)$, if the entries of
 each row of $D$ comprise all the elements of $I_m$, and for all $0 \le s < t \le n-1$, the lists of differences $\Delta_{ts} = \{d_{tj}-d_{sj} \ | \ 0 \le j \le m-1\}$ comprise all the elements of $I_m$.

 A set of $n$ {\em properly centered $m \times m$ permutation matrices} were introduced by Zhang and Tu in \cite{ZT94} for the construction of radar arrays. The equivalences among a perfect difference matrix, an additive sequence of permutations and a set of properly centered permutation matrices  were presented in \cite{CWF2009, GMS2010}.

\begin{lemma}\label{equiv}{\rm (\cite{CWF2009, GMS2010})} ~
An {\rm ASP}$(n,m)$, a {\rm PDM}$(n,m)$ and a set of $n$ properly centered $m\times m$
permutation matrices are all equivalent.
\end{lemma}

 These equivalences become important clues to the solution of existence problem of additive sequences of permutations for the reason that there are  more methods to
 handle difference matrices than those of for additive sequences of permutations or a set of properly centered permutation matrices as showed in \cite{GMS2010}.
 In the rest of this paper, we will use perfect difference matrix to represent these concepts.

Because of the importance both from a theoretical and an applied point of view,
perfect difference matrices (or their equivalent concepts) have also attracted the attention of many researchers
(see, for examples \cite{Ab84, AK81, KL78, Ma87, Tu80, Tu80-2, Tu82}). However, known constructions and existence results on perfect difference matrices are also very few.
It is known that there exist PDM$(3,m)$s for $m=5,7,13$-$201$ with $m$ odd \cite{GMS2010,WC2010}.

The purpose of this paper is to tackle these difficult problems of perfect difference matrices, perfect difference families and related codes.
In Section 2, some new recursive constructions of PDM$(3,m)$s are obtained by using standard incomplete perfect difference matrix, and it is proved that there exist PDM$(3,m)$s for $m\leq1000$ with some possible exceptions.
 In Section 3, new perfect difference families are investigated. By introducing
a special type of perfect difference family, {\em variable perfect difference family},
a class of $(g,\{3,4\},1)$ perfect difference families are obtained. In Section 4,
new perfect strict optical orthogonal codes (SOOCs) are constructed by using perfect difference families. Conclusions are given in Section 5.

\section{New Perfect Difference Matrices}

For a given $m=2r+1$, $r\geq1$, a PDM$(n,m),n\geq2$ implies a PDM$(t,m)$ with any $2\leq t \leq n$.
Let $N(m)=$max$\{n| {\rm a \ PDM}(n,m)\ {\rm exists}\}$,
it is  important to determine $N(m)$. For convenience, in this paper, we denote by $[a,b]$ the
set of integers $c$ such that $a\leq c\leq b$. Many researchers have studied the bounds for
 $N(m)$ \cite{Ab84,AK81,CWF2009,GMS2010,KL78,Ma87,Tu80,Tu80-2,Tu82,WC2010}. It is proved that $N(m)\geq 2$ for each odd $m\geq3$ in  \cite{Ab84}
and $N(m)\geq 3$ for each odd $m\in[5,201]\setminus\{9,11\}$ in  \cite{GMS2010,WC2010}, and there does not exist a PDM$(3,m)$
for $m=9,11$. For $N(m)=4$,  very scarce result is known \cite{CWF2009}.
It is proved  that $N(m)\leq m-1$ \cite{Tu80}.


In this section, we will construct more PDM$(3,m)$s by using recursive constructions. From
Lemma \ref{equiv}, this is equivalently to construct more  sets of $3$ properly centered $m\times m$
permutation matrices or more {\rm ASP}$(3,m)$s.

\begin{theorem}\label{multiplePDM1}{\rm (\cite{Tu80})} ~
Let $(X^{j}=(x_{1}^{j},...,x_{2r+1}^{j})|j=1,...,n)$ and $(Y^{j}=(y_{1}^{j},...,y_{2s+1}^{j})|j=1,...,n)$ be a {\rm PDM}$(n,2r+1)$ and a {\rm PDM}$(n,2s+1)$,
 respectively. Then $(Z^{j}=(z_{1}^{j},...,z_{(2r+1)(2s+1)}^{j})|j=1,...,n)$, where
$z_{(i-1)(2s+1)+h}^{j}=(2s+1)x_{i}^{j}+y_{h}^{j}$, $1\leq i\leq2r+1$,$1\leq h\leq2s+1$\\
is a {\rm PDM}$(n,(2r+1)(2s+1))$.
\end{theorem}

Abrham \cite{Ab84} developed the following direct constructions for additive sequences of
permutations from  perfect systems of difference sets.
\begin{lemma}\label{PSDS to PDM}{\rm (\cite{Ab84})} ~
If there exists a $(t,4,c)$-{\rm PSDS}, then there exists a {\rm ASP}$(3,12t)$ with the basis containing the elements $\{\pm c,\pm(c+1),...,\pm(6t+c-1)\}$.
~Furthermore. if $c=1$, then one can construct a {\rm ASP}$(3,12t+1)$ with the basis containing the elements $\{0\}\cup\{\pm1,\pm2,...,\pm6t\}$.
\end{lemma}

Ge, Ling and Miao \cite{GLM08} presented a general construction of PDM$(n,m)$ via $(t,K,1)$-PSDS and TD($n,k$).

In this section, we will present several new recursive constructions of perfect difference matrices.
Perfect difference matrix with a regular hole is required. For $r\in \Z^{+}$, let $rI_{h}$=$\{ri|i\in I_{h}\}$.
An $n\times(m-h)$ matrix $D_H=(\delta_{ij})$ with entries from $I_{m} \setminus (lI_{h})$,
 where $lI_{h}\subseteq I_{m}$ for some $l\in \Z^{+}$, is called a standard incomplete perfect difference matrix
  with a regular hole $lI_{h}$, denoted briefly by SIPDM$(n,m,h,l)$, if the entries of each row of $D_H$
  comprise all the elements of $I_m \setminus (lI_h)$, and
for any $0 \le s < t \le n-1$, the difference set ${\Delta}_{ts} = \{{\delta}_{tj}-{\delta}_{sj} \ |  \ 0 \le j \le m-h-1\} = I_{m} \setminus
(lI_h)$ holds. When $h=1$, we can drop the letter  $l$ from the notation SIPDM$(n,m,1,l)$
since for any $l\in {\mathbf Z^{+}}$, we always have $lI_1=\{0\}$.
Clearly, by adding the column vector $(0,\ldots,0)^T$ to an SIPDM$(n,m,1)$, we immediately obtain a PDM$(n,m)$.

\begin{lemma}\label{SIPDM}
If there exists a $(v,4,1)$-{\rm PDP} which cover $H$ with $|H|=6h$, then there exists a
$3\times 12h$ matrix $D_{H}$, which is an {\rm SIPDM} based on $H\bigcup(-H)$.
\end{lemma}
\begin{proof}
Let ${\cal B} = \{B_1,B_2,\ldots,B_h\}$ be a $(v,4,1)$-{\rm PDP} which cover $H$,
where $B_i=\{0,a_i,b_i,c_i\}$ for $1\leq i\leq h$, let
\[D_i=\left(\small
\begin{array}{cccccccccccccccccccc}
-c_i & a_i-c_i & -b_i & b_i-c_i & a_i-b_i & -a_i & a_i     & b_i-a_i & c_i-b_i & b_i     & c_i-a_i & c_i \\
-b_i & a_i     & -a_i & b_i-a_i & a_i-c_i & -c_i & a_i-b_i & b_i     & c_i     & b_i-c_i & c_i-b_i &
 c_i-a_i \\
-a_i & a_i-b_i & -c_i & b_i     & a_i     & -b_i & a_i-c_i & b_i-c_i & c_i-a_i & b_i-a_i & c_i     &
 c_i-b_i
\end{array}
\right),
\]
and $D_{H}=(D_1,D_2,\ldots,D_h)$. Combining Theorem 10 of \cite{Ma87} and Theorem 3.10 of \cite{GMS2010},
one can see that $D_{H}$ is an {\rm SIPDM} based on $H\bigcup(-H)$.
\end{proof}

We will employ SIPDMs to give  new recursive constructions for PDM$(n,m)$s. It is easy to check that
the ASP$(3,12t)$s in Lemma~\ref{PSDS to PDM} are SIPDMs.
More generally, if there exists a PDF$(12t+1,4,1)$, then there exists a SIPDM$(3,12t+1,1)$.
 In the following, we will investigate SIPDM$(3,m,1)$s.
 \begin{lemma}\label{m=5,7}
 There  exist no {\rm SIPDM}$(3,m,1)$s for $m=5,7$.
 \end{lemma}
 \begin{proof}  \label{m=5,7pf}
If $D_{H}^{(1)}=(\delta_{ij})$ is an SIPDM$(3,5,1)$, without loss of generality, we can assume that
$\delta_{11}=-2$. Then for
$\delta_{21}$, since $\delta_{21}-\delta_{11}\in I_{5}\setminus \{0\}$,  it holds $\delta_{21}=-1$.
Similarly, we obtain $\delta_{31}=-1$, it conflicts with $\delta_{31}-\delta_{21}\in I_{5}\setminus \{0\}$.

If $D_{H}^{(2)}=(\delta_{ij})$ is an SIPDM$(3,7,1)$, without loss of generality, we can assume that
$\delta_{11}=-3$, $\delta_{12}=-2$ and $\delta_{21}<\delta_{31}$.
 Since $\delta_{21}-\delta_{11}, \delta_{31}-\delta_{11} \in I_{7}\setminus \{0\}$, we obtain that $\delta_{21}=-2,\delta_{31}=-1$.
 Then we have $\delta_{21}-\delta_{11}=1$, it holds that $\delta_{22}-\delta_{12}\in I_{7}\setminus \{0,1\}$, therefore,
$\delta_{22}=-3$ or $1$. Both of them mean $\delta_{32}=-1$, which conflicts with $\delta_{31}=-1$.
\end{proof}

\begin{lemma} \label{SPDM1}
There exists an {\rm SIPDM}$(3,m,1)$ for each odd integer $5\leq m<200$ except for $m=5$, $7$,
$9$,$11$ and except possibly for $m\in\{15$, $21$, $27$, $29$, $35$, $47$, $51$, $53$, $59$,
$63$, $71$, $75$, $83$, $87$, $95\}$.
\end{lemma}
\begin{proof} \label{SPDM1pf}
From Lemma \ref{m=5,7}, there does not exist an SIPDM$(3,m,1)$ for $m=5,7$. Since
there does not exist a PDM$(3,m)$ for $m=9,11$, then so does for an SIPDM$(3,m,1)$ for $m=9,11$.
Since a $(12t+1,4,1)$-PDF can produce an SIPDM$(3,12t+1,1)$, then, one can obtain SIPDM$(3,m,1)$s for
$m=13$, $49$, $61$, $73$, $85$, $97$, $109$, $121$, $133$, $145$, $157$, $169$, $181$, $193$ from the $(12t+1,4,1)$-PDFs for $t=1,4$-$16$ (\cite{AB07}). An SIPDM$(3,17,1)$ is given in \cite{Ab84}.
SIPDM$(3,m,1)$s for $m=57$, $65$, $69$, $77$, $81$, $89$, $93$, $101$, $105$, $113$, $117$, $125$,
$129$, $137$, $141$, $149$ are from  \cite{GLM08}, SIPDM$(3,m,1)$s for
$m=31$, $33$, $37$, $39$, $43$, $55$, $67$, $79$ are from \cite{WC2010}, and  SIPDM$(3,m,1)$s for
$m=23$, $25$, $41$, $45$, $91$, $99$, $103$, $107$, $111$, $115$, $119$, $123$, $127$, $131$,
$135$, $139$, $143$, $147$, $151$, $153$, $155$, $159$, $161$, $163$, $165$, $167$, $169$,
$171$, $173$, $175$, $177$, $179$, $183$, $185$, $187$, $189$, $191$, $195$, $197$, $199$
are from \cite{GMS2010}. An SIPDM$(3,19,1)$ is given below.
\[\left(
\begin{array}{rrrrrrrrrrrrrrrrrrrrrrrrrrrrrrrr}
-9 & -7 & -6  & -3 & -1 & 8  \\
-8 & 1 & -2 & 6 & 5 & 2  \\
-5 & -4 & 3 & 4 & 7 & 9  \\
\end{array}
\right).
\]
In the above matrix, each column $(a,b,c)^{T}$ represents three columns $(a,b,c)^{T},(b,c,a)^{T}$ and $(c,a,b)^{T}$ in the SIPDM$(3,19,1)$.
\end{proof}

Next, we give some new recursive constructions of PDM$(3,m)$s by using SIPDMs.
\begin{theorem} \label{PDM Construction1}
If there exist a {\rm PDM}$(3,m_{1})$, a {\rm PDM}$(3,2m_{1}-1)$ and an {\rm SIPDM}$(3,m_{2},1)$, then there exists a {\rm PDM}$(3,m)$ with $m=m_{1}(m_{2}+1)-1$.
If the {\rm PDM}$(3,2m_{1}-1)$ is an {\rm SIPDM}, then the {\rm PDM}$(3,m)$ is also an {\rm SIPDM}.
\end{theorem}
\begin{proof}
Let $D=(a_{li})$, $1\leq l \leq3$, $1\leq i\leq m_{1}$, be a PDM$(3,m_{1})$ with $m_{1}=2r_{1}+1$,
and $D_{H}^{(1)}=(b_{lj})$, $1\leq l \leq3$, $1\leq j\leq m_{2}-1$, be an SIPDM$(3,m_{2},1)$ with $m_{2}=2r_{2}+1$.
Let $D_{H}=(d_{ls})$ where $d_{ls}=(m_{2}+1)a_{li}+b_{lj}$, $1\leq l \leq3$, $s=(m_{2}-1)(i-1)+j$,
$1\leq i\leq m_{1}$, $1\leq j\leq m_{2}-1$, then
$D_{H}$ is an SIPDM of basis $I_{m_{1}(m_{2}+1)-1}\setminus (r_{2}+1)I_{2m_{1}-1}$.
Multiplying a PDM$(3,2m_{1}-1)$ by $r_{2}+1$, we can obtain a $3\times (2m_{1}-1)$ matrix $M$,
and $M$ is a new PDM based on $(r_{2}+1)I_{2m_{1}-1}$, then $(D_{H},M)$
forms a PDM$(3,m_{1}(m_{2}+1)-1)$ of basis $I_{m_{1}(m_{2}+1)-1}$.
\end{proof}

Similar to the proof of Theorem \ref{PDM Construction1}, one can obtain the following result.

\begin{theorem} \label{PDM Construction1-1}
If there exist a {\rm PDM}$(3,m_{1})$, a {\rm PDM}$(3,2m_{1}+1)$ and an {\rm SIPDM}$(3,m_{2},1)$,
then there exists a {\rm PDM}$(3,m)$ with $m=m_{1}(m_{2}+1)+1$.
If the {\rm PDM}$(3,2m_{1}+1)$ is an {\rm SIPDM}, then the {\rm PDM}$(3,m)$ is also an {\rm SIPDM}.
\end{theorem}

\begin{theorem}\label{PDM Construction2}
Let $t\geq \frac{m_{1}+1}{12}$, $m_{1}=2r_{1}+1$. If there exist a {\rm PDM}$(3,m_{1})$, an
{\rm SIPDM}$(3,m_{2},1)$ and a $(2x+12t+1,4,1)$-{\rm PDP} with $t$ blocks which cover $[1,r_{1}]\cup[x+r_{1}+1,x+6t]$,
then there exists an {\rm SIPDM}$(3,m,1)$ with $m=m_{1}(m_{2}-1)+12t+1$.
\end{theorem}
\begin{proof}
Let $D=(a_{li})$, $1\leq l \leq3$, $1\leq i\leq m_{1}$, be a PDM$(3,m_{1})$ with $m_{1}=2r_{1}+1$,
and $D_{H}^{(1)}=(b_{lj})$, $1\leq l \leq3$, $1\leq j\leq m_{2}-1$, be an SIPDM$(3,m_{2},1)$ with $m_{2}=2r_{2}+1$.
Let $D_{H}^{(2)}=(d_{ls})$ where $d_{ls}=m_{1}b_{lj}+a_{li}$, $1\leq l \leq3$, $s=m_{1}(j-1)+i$,
$1\leq i\leq m_{1}$, $1\leq j\leq m_{2}-1$,
then $D_{H}^{(2)}$ is an SIPDM of basis $I_{m_{1}m_{2}}\setminus I_{m_{1}}$.
Let $H=[1,r_{1}]\cup[m_{1}r_{2}+r_{1}+1,m_{1}r_{2}+6t]$.
Setting $x=m_{1}r_{2}$, since there exists a $(2x+12t+1,4,1)$-PDP with $t$ blocks which cover $[1,r_{1}]\cup[x+r_{1}+1,x+6t]$, by Lemma~\ref{SIPDM},
we can obtain an SIPDM based on $H\bigcup(-H)$. Denote this SIPDM as $D_{H}^{(3)}$. Let $D_{H}=(D_{H}^{(3)},D_{H}^{(2)})$,
we obtain an SIPDM$(3,m,1)$ with $m=m_{1}(m_{2}-1)+12t+1$ which is based on $I_{m}$.
\end{proof}

\begin{corollary} \label{PDMCon2result}
If there exists an {\rm SIPDM}$(3,m_{2},1)$, then there exist {\rm SIPDM}$(3,m,1)$ for
$m=5m_{2}$ $+8,19m_{2}+30,29m_{2}+44,29m_{2}+56,37m_{2}+48,39m_{2}+46,45m_{2}+52$.
\end{corollary}
\begin{proof}
Obviously, $\{0,1,x+4,x+6\}$ is a $(2x+13,4,1)$-PDP with $1$ block which cover
$[1,2]\cup[x+3,x+6]$. Applying Lemma~\ref{PDM Construction2} with $r_1=2, t=1$, one can obtain an SIPDM$(3,5m_{2}+8,1)$.
From Theorem 11 in \cite{Ma87} and Lemma 4.2 of \cite{GMS2010}, there exists an $(2x+12t+1,4,1)$-PDP
with $t$ blocks which cover $[1,r_{1}]\cup[x+r_{1}+1,x+6t]$
for each $(r_{1},t)\in\{(9,4),(14,6),(14,7),(18,7),(19,7),(22,8)\}$, the conclusion comes from
Theorem~\ref{PDM Construction2}.
\end{proof}

\begin{lemma}\label{PDMCon3}
Let $t\geq \frac{m_{2}+1}{12}$, $m_{2}=2r_{2}+1$. If there exist an {\rm SIPDM}$(3,m_{2},1)$
and a $(2x+12t+1,4,1)$-{\rm PDP} with $t$ blocks which cover $[1,r_{2}]\cup[x+r_{2}+1,x+6t]$,
then there exists a {\rm PDM}$(3,m)$ with $m=12(m_{2}+t)+15$.
\end{lemma}
\begin{proof}
Let $D_H=(b_{lj})$, $1\leq l \leq3$, $1\leq j\leq m_{2}-1$, be an SIPDM$(3,m_{2},1)$ with $m_{2}=2r_{2}+1$.
Let $B_{j}=\{0,2(m_{2}+1)+b_{1j},5(m_{2}+1)+b_{2j},6(m_{2}+1)+b_{3j}\}$,
$j=1,2,...,m_{2}-1$, where $\{\{0,2,5,6\}\}$ is a $(13,4,1)$-PDF, then
${\cal B}= \{B_{j}|j=1,2,...,m_{2}-1\}$ is a
$(26r_{2}+27,4,1)$-PDP which covers $[r_{2}+1,13(r_{2}+1)]\setminus \{(r_{2}+1)i|1\leq i\leq 13\}$.
Let $x=12r_{2}+13$, and applying the $(2x+12t+1,4,1)$-PDP which covers $[1,r_{2}]\cup[x+r_{2}+1,x+6t]$,
one can obtain a new $(24r_{2}+12t+27,4,1)$-PDP which covers
$[1,12r_{2}+6t+13]\setminus \{(r_{2}+1)i|1\leq i\leq 13\}$.
By Lemma~\ref{SIPDM}, we can obtain an SIPDM based on $I_m\setminus ((r_{2}+1)I_{27})$.
Denote this SIPDM as $D_{H}$.
Multiplying a PDM$(3,27)$ by $r_{2}+1$, we can obtain a $3\times 27$ matrix $M$,
which is a new PDM based on $(r_{2}+1)I_{27}$, then $(D_{H},M)$
forms a PDM$(3,m)$ with $m=2(12r_{2}+6t+13)+1=12(m_{2}+t)+15$.
\end{proof}

  \begin{lemma}\label{p}
There exists a {\rm PDM}$(3,m)$ for each $m\in$ $S_1=\{p|211\leq p\leq 997,$ $p$ is a prime$\}$
except possibly for $m\in E_1=\{227$, $229$, $241$, $251$, $257$, $263$, $277$, $317$,
$331$, $347$, $367$, $373$, $383$, $397$, $401$, $431$, $439$, $587$, $617$, $641$,
$677$, $709$, $719$, $757$, $877$, $947$, $971$, $977$, $997\}$.
 \end{lemma}
\begin{proof} \ For each $m\in S_1\setminus E_1$, the corresponding $m_1s, m_2$s in Theorems \ref{PDM Construction1}, \ref{PDM Construction1-1},
and $m_2$s in Corollary \ref{PDMCon2result} are listed in Table 1.
\end{proof}

\begin{lemma}\label{3p}
 There exists a {\rm PDM}$(3,m)$ for each $m\in$ $S_2=\{3p|71\leq p\leq331,$ $p$ is a prime$\}$
except possibly for $m\in E_2=\{219$, $249$, $291$, $303$, $327$, $471$,
$591$, $669$, $717$, $789$, $829$, $831$, $879\}$.
\end{lemma}
\begin{proof} \ For each $m\in S_2\setminus E_2$, the corresponding $m_1$s, $m_2$s in Theorems
\ref{PDM Construction1}, \ref{PDM Construction1-1},
and $m_2$s in Corollary \ref{PDMCon2result} are listed in Table 2.
\end{proof}

 \begin{lemma}\label{9p}
 There exists a {\rm PDM}$(3,m)$ for each $m\in$ $S_3=\{9p|23\leq p\leq109,$ $p$ is a prime$\}$
except possibly for $m\in E_3=\{207$, $369$, $387$, $423$, $639\}$.
\end{lemma}
\begin{proof} \ For each $m\in S_3\setminus E_3$, the corresponding $m_1$s, $m_2$s in Theorems
\ref{PDM Construction1}, \ref{PDM Construction1-1},
and $m_2$s in Corollary \ref{PDMCon2result} are listed in Table 3.
\end{proof}

\begin{lemma}\label{11p}
 There exists a {\rm PDM}$(3,m)$ for each $m\in$ $S_4=\{11p|19\leq p\leq89,$ $p$ is a prime$\}$
except possibly for $m=319$.
\end{lemma}
\begin{proof} \ For each $m\in S_4\setminus \{319\}$, the corresponding $m_1$s, $m_2$s in Theorems
\ref{PDM Construction1}, \ref{PDM Construction1-1},
and $m_2$s in Corollary \ref{PDMCon2result} are listed in Table 4.
\end{proof}

  \begin{center} \footnotesize    {\bf Table 1} Corresponding $m_1$s and $m_2$s in Lemma \ref{p}
  \end{center}
\begin{center}\footnotesize
\begin{tabular}{|l|l|l|l|l|l|l|l|l|l|l|l|}
\hline
$m$ &$m_1$&$m_2$& Reference  &  $m$  &  $m_1$  &  $m_2$ &  Reference
&  $m$  &  $m_1$  &  $m_2$ &  Reference   \\
\hline
211 & 15 & 13 & Theorem~\ref{PDM Construction1-1} & 223 & 7  & 31 & Theorem~\ref{PDM Construction1}
& 233 & 13 & 17 & Theorem~\ref{PDM Construction1}\\
239 & 17 & 13 & Theorem~\ref{PDM Construction1-1} & 269 & 15 & 17 & Theorem~\ref{PDM Construction1}
& 271 & 15 & 17 & Theorem~\ref{PDM Construction1-1}\\
281 & 7  & 39 & Theorem~\ref{PDM Construction1-1} & 283 &    & 55 & Corollary~\ref{PDMCon2result}
& 293 & 21 & 13 & Theorem~\ref{PDM Construction1}\\
307 & 17 & 17 & Theorem~\ref{PDM Construction1-1} & 311 & 13 & 23 & Theorem~\ref{PDM Construction1}
& 313 & 13 & 23 & Theorem~\ref{PDM Construction1-1}\\
337 & 13 & 25 & Theorem~\ref{PDM Construction1}   & 349 & 25 & 13 & Theorem~\ref{PDM Construction1}
& 353 &    & 69 & Corollary~\ref{PDMCon2result}\\
359 & 15 & 23 & Theorem~\ref{PDM Construction1}   & 379 & 19 & 19 & Theorem~\ref{PDM Construction1}
&389 & 15 & 25 & Theorem~\ref{PDM Construction1}\\
409 & 17 & 23 & Theorem~\ref{PDM Construction1-1} & 419 & 21 & 19 & Theorem~\ref{PDM Construction1}
& 421 & 21 & 19 & Theorem~\ref{PDM Construction1-1}\\
433 & 31 & 13 & Theorem~\ref{PDM Construction1}   & 443 & 17 & 25 & Theorem~\ref{PDM Construction1-1}
& 449 & 25 & 17 & Theorem~\ref{PDM Construction1}\\
457 & 19 & 23 & Theorem~\ref{PDM Construction1-1} & 461 & 7  & 65 & Theorem~\ref{PDM Construction1}
& 463 & 7 & 65 & Theorem~\ref{PDM Construction1-1}\\
467 &    & 23 & Corollary~\ref{PDMCon2result}     & 479 & 15 & 31 & Theorem~\ref{PDM Construction1}
& 487 & 27 & 17 & Theorem~\ref{PDM Construction1-1}\\
491 & 35 & 13 & Theorem~\ref{PDM Construction1-1} & 499 & 25 & 19 & Theorem~\ref{PDM Construction1}
& 503 & 21 & 23 & Theorem~\ref{PDM Construction1}\\
509 & 15 & 33 & Theorem~\ref{PDM Construction1}   & 521 & 29 & 17 & Theorem~\ref{PDM Construction1}
& 523 & 29 & 17 & Theorem~\ref{PDM Construction1-1}\\
541 & 27 & 19 & Theorem~\ref{PDM Construction1-1} & 547 & 39 & 13 & Theorem~\ref{PDM Construction1-1}
& 557 & 31 & 17 & Theorem~\ref{PDM Construction1}\\
563 &    & 111& Corollary~\ref{PDMCon2result}     & 569 & 15 & 37 & Theorem~\ref{PDM Construction1}
& 571 & 15 & 37 & Theorem~\ref{PDM Construction1-1}\\
577 & 17 & 33 & Theorem~\ref{PDM Construction1}   & 593 & 33 & 17 & Theorem~\ref{PDM Construction1}
& 599 & 25 & 23 & Theorem~\ref{PDM Construction1}\\
601 & 25 & 23 & Theorem~\ref{PDM Construction1-1} & 607 & 19 & 31 & Theorem~\ref{PDM Construction1}
& 613 &    & 121& Corollary~\ref{PDMCon2result}\\
619 & 31 & 19 & Theorem~\ref{PDM Construction1}   & 631 & 45 & 13 & Theorem~\ref{PDM Construction1-1}
& 643 & 7  & 91 & Theorem~\ref{PDM Construction1}\\
647 & 27 & 23 & Theorem~\ref{PDM Construction1}   & 653 &    & 129& Corollary~\ref{PDMCon2result}
& 659 & 33 & 19 & Theorem~\ref{PDM Construction1}\\
661 & 33 & 19 & Theorem~\ref{PDM Construction1-1} & 673 & 21 & 31 & Theorem~\ref{PDM Construction1-1}
& 677 &    & 17 & Corollary~\ref{PDMCon2result}\\
683 &    & 135& Corollary~\ref{PDMCon2result}     & 691 & 15 & 45 & Theorem~\ref{PDM Construction1-1}
& 701 & 35 & 19 & Theorem~\ref{PDM Construction1-1}\\
727 & 13 & 55 & Theorem~\ref{PDM Construction1}   & 733 &    & 145& Corollary~\ref{PDMCon2result}
& 739 & 37 & 19 & Theorem~\ref{PDM Construction1}\\
743 & 31 & 23 & Theorem~\ref{PDM Construction1}   & 751 & 15 & 49 & Theorem~\ref{PDM Construction1-1}
&761 & 19 & 39 & Theorem~\ref{PDM Construction1-1}\\
769 & 55 & 13 & Theorem~\ref{PDM Construction1}   & 773 & 43 & 17 & Theorem~\ref{PDM Construction1}
&787 &    & 19 & Corollary~\ref{PDMCon2result}\\
797 & 57 & 13 & Theorem~\ref{PDM Construction1}   & 809 & 45 & 17 & Theorem~\ref{PDM Construction1}
& 811 & 45 & 17 & Theorem~\ref{PDM Construction1-1}\\
821 & 41 & 19 & Theorem~\ref{PDM Construction1-1} & 823 &    & 163& Corollary~\ref{PDMCon2result}
& 827 & 59 & 13 & Theorem~\ref{PDM Construction1-1}\\
839 & 21 & 39 & Theorem~\ref{PDM Construction1}   & 853 & 61 & 13 & Theorem~\ref{PDM Construction1}
& 857 & 33 & 25 & Theorem~\ref{PDM Construction1}\\
859 & 43 & 19 & Theorem~\ref{PDM Construction1}   & 863 & 27 & 31 & Theorem~\ref{PDM Construction1}
& 881 & 49 & 17 & Theorem~\ref{PDM Construction1}\\
883 & 13 & 67 & Theorem~\ref{PDM Construction1}   & 887 & 37 & 23 & Theorem~\ref{PDM Construction1}
& 907 &    & 19 & Corollary~\ref{PDMCon2result}\\
911 & 65 & 13 & Theorem~\ref{PDM Construction1-1} & 919 & 23 & 39 & Theorem~\ref{PDM Construction1}
& 929 & 29 & 31 & Theorem~\ref{PDM Construction1-1}\\
937 & 67 & 13 & Theorem~\ref{PDM Construction1}   & 941 & 47 & 19 & Theorem~\ref{PDM Construction1-1}
& 953 & 53 & 17 & Theorem~\ref{PDM Construction1}\\
967 & 23 & 41 & Theorem~\ref{PDM Construction1-1} & 983 & 41 & 23 & Theorem~\ref{PDM Construction1}
& 991 & 31 & 31 & Theorem~\ref{PDM Construction1}\\
\hline
\end{tabular}
\end{center}



  \begin{center} \footnotesize
  {\bf Table 2} Corresponding $m_1$s and $m_2$s in Lemma \ref{3p}
  \end{center}
\begin{center}\footnotesize
\begin{tabular}{|l|l|l|l|l|l|l|l|l|l|l|l|}
\hline
$m$ &$m_1$&$m_2$& Reference  &  $m$  &  $m_1$  &  $m_2$ &  Reference
&  $m$  &  $m_1$  &  $m_2$ &  Reference   \\
\hline
213 &    & 41 & Corollary~\ref{PDMCon2result}     & 237 & 17 & 13 & Theorem~\ref{PDM Construction1}
& 267 & 19 & 13 & Theorem~\ref{PDM Construction1-1}\\
309 & 7  & 43 & Theorem~\ref{PDM Construction1-1} & 321 & 23 & 13 & Theorem~\ref{PDM Construction1}
& 339 & 17 & 19 & Theorem~\ref{PDM Construction1}\\
381 & 19 & 19 & Theorem~\ref{PDM Construction1-1} & 393 & 7  & 55 & Theorem~\ref{PDM Construction1-1}
& 417 & 13 & 31 & Theorem~\ref{PDM Construction1-1}\\
453 &    & 89 & Corollary~\ref{PDMCon2result}     & 489 & 35 & 13 & Theorem~\ref{PDM Construction1}
& 501 & 25 & 19 & Theorem~\ref{PDM Construction1-1}\\
519 & 13 & 39 & Theorem~\ref{PDM Construction1}   & 537 &    & 17 & Corollary~\ref{PDMCon2result}
&543 & 17 & 31 & Theorem~\ref{PDM Construction1}   \\
573 & 41 & 13 & Theorem~\ref{PDM Construction1}   & 579 & 29 & 19 & Theorem~\ref{PDM Construction1}
&597 & 23 & 25 & Theorem~\ref{PDM Construction1}   \\
633 &    & 125& Corollary~\ref{PDMCon2result}       & 681 & 17 & 39 & Theorem~\ref{PDM Construction1-1}
&687 & 49 & 13 & Theorem~\ref{PDM Construction1-1} \\
699 & 35 & 19 & Theorem~\ref{PDM Construction1}   & 723 & 19 & 37 & Theorem~\ref{PDM Construction1-1}
&753 & 29 & 25 & Theorem~\ref{PDM Construction1}   \\
771 & 55 & 13 & Theorem~\ref{PDM Construction1-1} & 807 & 31 & 25 & Theorem~\ref{PDM Construction1-1}
&813 &   & 161& Corollary~\ref{PDMCon2result}      \\
843 &    & 167& Corollary~\ref{PDMCon2result}     & 849 & 25 & 33 & Theorem~\ref{PDM Construction1}
&921 & 23 & 39 & Theorem~\ref{PDM Construction1-1} \\
933 &    & 185& Corollary~\ref{PDMCon2result}     & 939 & 47 & 19 & Theorem~\ref{PDM Construction1}
&951 & 17 & 55 & Theorem~\ref{PDM Construction1}   \\
993 & 71 & 13 & Theorem~\ref{PDM Construction1}   &     &    &    &
&     &    &    &    \\
\hline
\end{tabular}
\end{center}

 \newpage

  \begin{center}\footnotesize
  {\bf Table 3} Corresponding $m_1$s and $m_2$s in Lemma \ref{9p}
  \end{center}
\begin{center}\footnotesize
\begin{tabular}{|l|l|l|l|l|l|l|l|l|l|l|l|}
\hline
$m$ &$m_1$&$m_2$& Reference  &  $m$  &  $m_1$  &  $m_2$ &  Reference
&  $m$  &  $m_1$  &  $m_2$ &  Reference   \\
\hline
261 & 13 & 19 & Theorem~\ref{PDM Construction1-1} & 279 &  7 & 39 & Theorem~\ref{PDM Construction1}
& 333 &    & 65 & Corollary~\ref{PDMCon2result}\\
477 &  7 & 67 & Theorem~\ref{PDM Construction1-1} & 549 &    & 17 & Corollary~\ref{PDMCon2result}
&603 & 43 & 13 & Theorem~\ref{PDM Construction1-1} \\
657 & 47 & 13 & Theorem~\ref{PDM Construction1}   & 711 &    & 23 & Corollary~\ref{PDMCon2result}
&747 & 17 & 43 & Theorem~\ref{PDM Construction1}   \\
801 & 25 & 31 & Theorem~\ref{PDM Construction1-1} & 873 & 23 & 37 & Theorem~\ref{PDM Construction1}
&909 & 13 & 69 & Theorem~\ref{PDM Construction1}   \\
927 & 29 & 31 & Theorem~\ref{PDM Construction1}   & 963 & 37 & 25 & Theorem~\ref{PDM Construction1-1}
&981 & 49 & 19 & Theorem~\ref{PDM Construction1-1}\\
\hline
\end{tabular}
\end{center}


  \begin{center} \footnotesize
  {\bf Table 4} Corresponding $m_1$s and $m_2$s in Lemma \ref{11p}
  \end{center}
\begin{center}\footnotesize
\begin{tabular}{|l|l|l|l|l|l|l|l|l|l|l|l|}
\hline
$m$ &$m_1$&$m_2$& Reference  &  $m$  &  $m_1$  &  $m_2$ &  Reference
&  $m$  &  $m_1$  &  $m_2$ &  Reference   \\
\hline
209 & 15 & 13 & Theorem~\ref{PDM Construction1}   & 253 &    & 49 & Corollary~\ref{PDMCon2result}
& 341 & 19 & 17 & Theorem~\ref{PDM Construction1}\\
407 & 17 & 23 & Theorem~\ref{PDM Construction1}   & 451 & 25 & 17 & Theorem~\ref{PDM Construction1-1}
& 473 &    & 93 & Corollary~\ref{PDMCon2result}\\
517 & 37 & 13 & Theorem~\ref{PDM Construction1}   & 583 &    & 115& Corollary~\ref{PDMCon2result}
& 649 & 13 & 49 & Theorem~\ref{PDM Construction1}\\
671 & 21 & 31 & Theorem~\ref{PDM Construction1}   & 737 & 41 & 17 & Theorem~\ref{PDM Construction1}
& 781 & 23 & 33 & Theorem~\ref{PDM Construction1}\\
803 &    & 159& Corollary~\ref{PDMCon2result}     & 869 & 15 & 57 & Theorem~\ref{PDM Construction1}
& 913 &    & 181& Corollary~\ref{PDMCon2result}\\
979 & 49 & 19 & Theorem~\ref{PDM Construction1} &     &    &    &
&     &    &    &     \\
\hline
\end{tabular}
\end{center}

\begin{lemma}\label{PDM01}
There exists a {\rm PDM}$(3,m)$ for each  $m\in S_5=\{229$, $241$, $277$, $291$, $373$, $387$, $397$,
$411$, $447$, $471$, $531$, $591$, $639$, $709$, $757$, $829$, $831$, $877$, $879$, $997\}$.
\end{lemma}
\begin{proof}
There exists a $(12t+1,4,1)$-PDF for each $t\in S_P=\{19$, $20$, $23$, $31$, $33$, $59$,
$63$, $69$, $73$, $83\}$ from \cite{GMS2010}, then there exists a $(t,4,1)$-PSDS for
each $t\in S_P$. From Lemma~\ref{equiv} and Lemma~\ref{PSDS to PDM}, there exists a PDM$(3,m)$
for each $m\in\{229$, $241$, $277$, $373$, $397$, $709$, $757$, $829$, $877$, $997\}$.

From Theorem 11 in \cite{Ma87} and Lemma 4.2 in \cite{GMS2010}, there exists a $(2x+12t+1,4,1)$-PDP with $t$ blocks which covers $[1,r_{2}]\cup[x+r_{2}+1,x+6t]$ for each $(r_{2},t)\in\{(9,4)$, $(12,6)$,
$(13,6)$, $(14,7)$, $(15,7)$, $(17,8)$, $(19,9)$, $(21,9)$,
$(27,13)$, $(30,11)\}$, then from Lemma~\ref{PDMCon3}, there exists a PDM$(3,m)$ for each
$m\in\{291$, $387$, $411$, $447$, $471$, $531$, $591$, $639$, $831$, $879\}$.
\end{proof}

\begin{theorem}\label{PDM reuslt}
There exists a {\rm PDM}$(3,m)$ for each odd integer $5\leq m<1000$ except for $m=9,11$ and except
possibly for $m\in (E_1\cup E_2\cup E_3\cup \{243,297,319,363\})\setminus S_5=\{207$, $219$, $227$,
$243$, $249$, $251$, $257$, $263$, $297$, $303$, $317$, $319$, $327$, $331$, $347$, $363$,
$367$, $369$, $383$, $401$, $423$, $431$, $439$, $587$, $617$, $641$, $669$, $717$, $719$,
$789$, $947$, $971$, $977\}$.
\end{theorem}
\begin{proof}  \label{PDM result}
From the known results of PDMs and Theorem~\ref{multiplePDM1}, one can obtain PDM$(3,m)$s for
$m=3^{a_{1}}5^{a_{2}}7^{a_{3}}11^{a_{4}}\prod\limits_{i=6}^{99}(2i+1)^{b_{i}}$ with $a_{1}=0$
or $a_{1}\geq 3$, $a_{2},a_{3}\geq 0$, $a_{4}=0$ or $a_{4}\geq 2$, $b_{i}\geq 0$,
and $a_{1}+a_{2}+a_{3}+a_{4}+\sum\limits_{i=6}^{99}b_{i}\geq 1$
for $6 \leq i\leq 99$.
Let $V$ be the set of all the values of $m$  produced by the above product  and belongs to $[201,999]$.
Let $S=\{2i+1| 100\leq i\leq 499\}\setminus V$, then $S=\bigcup\limits_{i=1}^{4}S_{i} \cup\{243,297,363\}$.
The result comes from  Lemmas~\ref{p}-\ref{PDM01}.
\end{proof}

\section{Constructions for Perfect Difference Families}
In this section, we will  focus on the construction   of $(g,\{3,4\},1)$-PDFs.
For convenience, we will use type $3^s 4^t$ to denote a $(g,\{3,4\},1)$-PDF with $s$ blocks
of size 3 and $t$ blocks of size 4, where $s\geq0,t\geq1$.
For a $(g,\{3,4\},1)$-PDF of type $3^s 4^t$, it is easy to see that $g=6s+12t+1$,
thus $g\equiv1\pmod6, g\geq13$. Let $\theta=\frac{t}{s+t}$ be the  ratio of block size 4,
we will construct a class of $(g,\{3,4\},1)$-PDFs with
$\theta\geq\frac{1}{14}$ for each $g\equiv 1\pmod6$, $g\geq 13$.

To construct $(g,\{3,4\},1)$-PDFs, perfect Langford sequences will be used. Perfect Langford sequence is introduced in \cite{Simpson}.
 A sequence $\{c,c+1,...,c+s-1\}$ is a perfect Langford sequences (PLS$(s,c)$ for short) starting
  with $c$ if the set $\{1,2,...,2s\}$ can be arranged in disjoint pairs $(e_{i},f_{i})$, where $i=1,...,s$ such that $\{f_{1}-e_{1}, f_{2}-e_{2},..., f_{s}-e_{s}\}=\{c,c+1,...,c+s-1\}$.
The existence of perfect Langford sequences had been completely solved.
\begin{theorem}\label{PLS}{\rm (\cite{Simpson})}
A {\rm PLS}$(s,c)$ exists if and only if~\\
{\rm (1)} $s\geq2c-1$;~\\
{\rm (2)} $s\equiv0,1\pmod4$ when $c$ is odd;\ \ \ \
    $s\equiv0,3\pmod4$ when $c$ is even.
\end{theorem}

 For a given PLS$(s,c)$, let $A_{i}=\{0,e_{i}+z,f_{i}+z\}$ with  variable $z$, and $i=1,2,...,s$,
 then $\bigcup\limits_{i=1}^{s}\Delta(A_{i})=\{e_{i}+z,f_{i}+z,f_{i}-e_{i} | 1\leq i\leq s\}=[z+1,z+2s]\cup[c,c+s-1]$.
 Let $z=c+s-1$, all the blocks of $A_{i},1\leq i\leq s$, have size 3 and cover differences $[c,c+3s-1]$. For convenience,
 we will call $A_{i}$s the corresponding blocks from the PLS.

We also need the following construction of special perfect difference family which is called \emph{variable} $(12t+1,4,1)$-PDF.
\begin{definition} \label{vpdf-d}
For a $(12t+1,4,1)$-{\rm PDF}, $B_{i}=\{0,a_{i},b_{i},c_{i}\}$, $1\leq i\leq t$,
 let $B_{i}(x)=\{0,a_{i},x+b_{i},x+c_{i}\}$, $B(x)=\{B_{1}(x),B_{2}(x),...,B_{t}(x)\}$ with a variable $x$.
 If all the differences of $\Delta B(x)$ cover  $[1,2t]\cup [x+2t+1,x+6t]$, then it is called  a \emph{variable} $(12t+1,4,1)$-{\rm PDF}.
\end{definition}

\begin{lemma}\label{vpdf}{\rm (\cite{Sun2008})}~
There exists a variable $(12t+1,4,1)$-{\rm PDF} for each $6\leq t \leq17$.
\end{lemma}

Variable perfect difference families can be obtained from smaller one using perfect difference matrices.
The following  recursive construction is obtained by  using a PDM$(3,5)$.

\begin{lemma}\label{rc-vpdf}
If there exists a variable $(12t+1,4,1)$-{\rm PDF}, then there exists a variable $(12(5t+1)+1,4,1)$-{\rm PDF}.
\end{lemma}
\begin{proof}
For a given variable $(12t+1,4,1)$-PDF, $B_{i}=\{0,a_{i},b_{i},c_{i}\}$, $1\leq i\leq t$,
it covers differences $[1,2t]\cup[x+2t+1,x+6t]$ exactly once.
Let $(\alpha_1,\alpha_2,\alpha_3,\alpha_4,\alpha_5)$ be a PDM$(3,5)$,
where $\alpha_j=(m_{1j},m_{2j},m_{3j})^T$, $j=1,...,5$,
and $B'_{ij}=\{0,5a_{i}+m_{1j},x+5b_{i}+m_{2j},x+5c_{i}+m_{3j}\}$, $i=1,...,t$, $j=1,2,3,4,5$.
Then all the $B'_{ij}$s have $5t$ blocks and cover differences $[3,10t+2]\cup[x+10t+3,x+30t+2]$
exactly once.
Let $\mathcal{B'}=\{0,1,(x+30t)+4,(x+30t)+6\}\cup \{B'_{ij}|1\leq i\leq t, 1\leq j\leq 5\}$,
it is easy to find that $\mathcal{B'}$ is a variable $(12(5t+1)+1,4,1)$-PDF.
\end{proof}

For convenience, in the sequel,  a block of size $k$ will be called $k$-block, and a set of blocks
 of size $k$ will be called $k$-blocks.
A recursive construction on PDF$(g,\{3,4\},1)$s with $\theta\geq\frac{1}{14}$ is given below.

\begin{lemma}\label{recursion}
If there exists a $(g,\{3,4\},1)$-{\rm PDF} with type $3^{s}4^{t}$, $s\equiv0,1\pmod4$,
$\theta=\frac{t}{s+t}\geq \frac{1}{14}$, all the $4$-blocks form a variable $(12t+1,4,1)$-{\rm PDF},
 and all the $3$-blocks form a {\rm PLS}$(s,2t+1)$, then there exists a $(g_{1},\{3,4\},1)$-{\rm PDF}
  with type $3^{5s+y}4^{5t+1}$, the ratio of $4$-blocks  $\theta_{1}\geq \frac{1}{14}$,
  all the $4$-blocks form a variable $(60t+13,4,1)$-{\rm PDF} and all the
   $3$-blocks form a {\rm PLS}$(5s+y,10t+3)$, where $y=0,1,4,12,13$ if $s\equiv 0\pmod4$ and
$y=0,3,4,11,12$ if $s\equiv 1\pmod4$.

\end{lemma}
\begin{proof}
For a given $(g,\{3,4\},1)$-PDF, it is clear that $g=6s+12t+1.$
Since all $4$-blocks form a variable $(12t+1,4,1)$-PDF, they cover
differences $[1,2t]\cup[x+2t+1,x+6t]$ exactly once and each $4$-block has the form $\{0,a_{i},x+b_{i},x+c_{i}\}$, $i=1,...,t$.
By Lemma \ref{rc-vpdf}, there exists a
variable $(60t+13,4,1)$-PDF which cover $[1,10t+2]\cup[x+10t+3,x+30t+6]$. Let $\mathcal{B}_1$ be the set of all the blocks
of the variable $(60t+13,4,1)$-PDF, then $|\mathcal{B}_1|=5t+1$.

Let $Y_{0}=\{0,1,4,12,13\}$, $Y_{1}=\{0,3,4,11,12\}$, $s_{1}=5s+y$, $y\in Y_{0}$, if $s\equiv 0\pmod4$, and $y\in Y_{1}$, if $s\equiv 1\pmod4$.
Thus $s_1\equiv 0,1\pmod4$. Since a {\rm PLS}$(s,2t+1)$ exists, then $s\geq 4t+1$ from Theorem \ref{PLS}.
Since $s_{1}=5s+y\geq20t+5+y\geq20t+5$, then we get a PLS$(s_{1},10t+3)$, the set $\{1,2,...,2(5s+y)\}$
can be arranged in disjoint pairs $(d_{k},e_{k})$, such that $\{e_{1}-d_{1}, e_{2}-d_{2},..., e_{5s+y}-d_{5s+y}\}=\{10t+3,10t+4,...,10t+2+5s+y\}$,
Let $\mathcal{B}_2=\{0,d_{k}+10t+2+5s+y,e_{k}+10t+2+5s+y|1\leq k\leq 5s+y\}$, then $|\mathcal{B}_2|=5s+y$.
Let  $x=15s+3y,\ y\in Y_{i},\ i=0,1$,
 and $\mathcal{B}=\mathcal{B}_1\cup \mathcal{B}_2$.
Then $\mathcal{B}$ forms a $(g_{1},\{3,4\},1)$-PDF with with type $3^{5s+y}4^{5t+1}$.
The ratio $4$-blocks $\theta_{1}=\frac{5t+1}{5s+y+5t+1}$, it holds $\theta_{1}\geq \frac{1}{14}$.
This completes the proof.
\end{proof}

\begin{lemma}\label{5timesPDF}
If there exists a $(g,\{3,4\},1)$-{\rm PDF} with type $3^{s}4^{t}$,$s\equiv0,1\pmod4$,
the ratio of $4$-blocks $\theta\geq\frac{1}{14}$, all the $4$-blocks form a variable
$(12t+1,4,1)$-{\rm PDF},
  and all the $3$-blocks form a {\rm PLS}$(s,2t+1)$ for each $h=s+2t\in[2b,10b+9]$,
  then there exists a $(g_{1},\{3,4\},1)$-{\rm PDF} with type $3^{5s+y}4^{5t+1}$, the ratio of  $4$-blocks
   $\theta_{1}\geq\frac{1}{14}$, all the $4$-blocks form a variable $(12t_{1}+1,4,1)$-{\rm PDF},
    and all the $3$-blocks form a {\rm PLS}$(5s+y,10t+3)$ for each $h_{1}\in[10b+10,50b+51]$,
 where $s_{1}=5s+y, t_{1}=5t+1$, $h_{1}=s_{1}+2t_{1}$, $y$ is the same as in
Lemma \ref{recursion}.
\end{lemma}
\begin{proof}
For each $g=6h+1, h=s+2t\in[2b,10b+9]$, there exists a $(g,\{3,4\},1)$-{\rm PDF} with type $3^{s}4^{t}$. Then, from Lemma \ref{recursion},
 there exists a $(g_{1},\{3,4\},1)$-{\rm PDF} with type $3^{5s+y}4^{5t+1}$, where $y$ is the same as
in Lemma \ref{recursion} according to the remainder of $s$ module $4$, and he ratio of $4$-blocks  $\theta_{1}\geq \frac{1}{14}$. It is clear that
$g_1=6(5s+y)+12(5t+1)+1=6h_1+1$, and $h_{1}=s_{1}+2t_{1}=5s+y+2(5t+1)=5s+y+10t+2=5h+y+2$.
Let $A=\{5h+y+2| 2b\leq h\leq10b+9\}$, then the result can be obtained if we can prove that $[10b+10,50b+51]\subseteq A$.

Let $Y_{0}=\{0,1,4,12,13\}$, $Y_{1}=\{0,3,4,11,12\}$, $A_{0}=\{5h+y+2| 2b\leq h\leq10b+8, h\ even,\  y\in Y_{0}\}$, and
$A_{1}=\{5h+y+2, 2b+1\leq h\leq10b+9, h\ odd,\  y\in Y_{1}\}$,
then $A=A_{0}\bigcup A_{1}$. Note that

 $A_{0}=\{5h+2,5h+3,5h+6,5h+14,5h+15 |2b\leq h\leq10b+8, h=2d\}
=\{10d+2,10d+3,10(d+1)+4,10(d+1)+5,10d+6 |b\leq d\leq5b+4\}$,

$A_{1}=\{5h+2,5h+5,5h+6,5h+13,5h+14 |2b+1\leq h\leq10b+9, h=2d+1\}
=\{10(d+1)+1,10d+7,10(d+1)+8,10(d+1)+9,10d+10 |b\leq d\leq5b+4\}$.
Then, we have

 $A=A_{0}\bigcup A_{1}=\{10(d+1)+1,10d+2,10d+3,10(d+1)+4,10(d+1)+5,10d+6,10d+7,
10(d+1)+8,10(d+1)+9,10d+10 |b\leq d\leq5b+4\}
=\{10b+2,10b+3,10b+6,10b+7\}\bigcup[10b+10,50b+51]\bigcup\{50b+54,50b+55,50b+58,50b+59\}$.\\
Thus, $[10b+10,50b+51]\subseteq A$. This completes  the proof.
\end{proof}


\begin{lemma}\label{middleVPDF}
There exists a $(g,\{3,4\},1)$-{\rm PDF}  with $g=6h+1,h=s+2t$, the ratio of
$4$-blocks  $\theta\geq\frac{1}{14}$, where all the $4$-blocks form a
variable $(12t+1,4,1)$-{\rm PDF},  and all the $3$-blocks form a
{\rm PLS}$(s,2t+1)$ for each $h\in[43,240]\cup\{37,40,41,242,243,246,247,250,251,254,255\}$.
\end{lemma}
\begin{proof}\label{middleVPDFpf}
For a given variable $(12t+1,4,1)$-PDF, all $4$-blocks cover
differences $[1,2t]\cup[x+2t+1,x+6t]$ exactly once.
Let $s=4e$ with $t+1 \leq e\leq \lfloor \frac{13t}{4}\rfloor$ or $s=4e+1$ with
$t\leq e\leq \lfloor \frac{13t-1}{4}\rfloor$, there exists a PLS$(s,2t+1)$ by
Theorem \ref{PLS}, then all $3$-blocks cover differences $[2t+1,2t+3s]$ exactly once.
Let $x=3s$, combine the above variable PDF and the corresponding blocks from the PLS, we get a $(6h+1,\{3,4\},1)$-PDF
with $h=s+2t$, and $\frac{t}{s+t}\geq\frac{1}{14}$.
For $6\leq t\leq17$, let  $A_{t}$ be the set of $\{s+2t\}$, then

$A_{6}=\{4e+12| 7\leq e\leq19\}\cup\{4e+13|6\leq e\leq 19\}$. ~\\
\ \indent $A_{7}=\{4e+14|8\leq e\leq 22\}\cup \{4e+15|7\leq e\leq 22\}$.~\\
\ \indent $A_{8}=\{4e+16|9\leq e\leq 26\}\cup \{4e+17|8\leq e\leq 25\}$.~\\
\ \indent $A_{9}=\{4e+18|10\leq e\leq 29\}\cup \{4e+19|9\leq e\leq 29\}$.~\\
\ \indent $A_{10}=\{4e+20|11\leq e\leq 32\}\cup \{4e+21|10\leq e\leq 32\}$.~\\
\ \indent $A_{11}=\{4e+22|12\leq e\leq 35\}\cup \{4e+23|11\leq e\leq 35\}$.~\\
\ \indent $A_{12}=\{4e+24|13\leq e\leq 39\}\cup \{4e+25|12\leq e\leq 38\}$.~\\
\ \indent $A_{13}=\{4e+26|14\leq e\leq 42\}\cup \{4e+27|13\leq e\leq 42\}$.~\\
\ \indent $A_{14}=\{4e+28|15\leq e\leq 45\}\cup \{4e+29|14\leq e\leq 45\}$.~\\
\ \indent $A_{15}=\{4e+30|16\leq e\leq 48\}\cup \{4e+31|15\leq e\leq 48\}$.~\\
\ \indent $A_{16}=\{4e+32|17\leq e\leq 52\}\cup \{4e+33|16\leq e\leq 51\}$.~\\
\ \indent $A_{17}=\{4e+34|18\leq e\leq 55\}\cup \{4e+35|17\leq e\leq 55\}$.~\\
\ \indent Let $A=\bigcup\limits_{i=6}^{17}A_{i}$, then it is easy to see that  $A=[43,240]\cup\{37,40,41,242,243,246,247,250,251,$ $254,255\}$.
The desired $(6h+1,\{3,4\},1)$-PDF is obtained from $A$.
\end{proof}

\begin{theorem}\label{all pdf new}
For each $h\in\{37,40,41,43\}$ or $h\geq 44$, let $g=6h+1$, $h=s+2t$, for some $s$, $t$,
there exists a $(g,\{3,4\},1)$-{\rm PDF} with type $3^s4^t$, the ratio of  $4$-blocks $\theta\geq\frac{1}{14}$,
all the $4$-blocks form a variable $(12t+1,4,1)$-{\rm PDF} and all the $3$-blocks form a {\rm PLS}$(s,2t+1)$.
\end{theorem}
\begin{proof}
For each $h\in\{37,40,41,43\}$, the conclusion comes from  Lemma \ref{middleVPDF}.
For $b=22$, and each $h\in[2b,10b+9]=[44,229]$,  there exists a desired  $(6h+1,\{3,4\},1)$-{\rm PDF} from  Lemma \ref{middleVPDF},
then from Lemma \ref{5timesPDF},  one can get a $(6h+1,\{3,4\},1)$-{\rm PDF} for each $44\leq h\leq 1151$.
The conclusion can be obtained by repeatedly using Lemma \ref{5timesPDF}.
\end{proof}

\begin{lemma}\label{smallMPDF} 
For each $h\in[2,36]\cup\{38,39,42\}$, there exists
a $(g,\{3,4\},1)$-{\rm PDF} with $g=6h+1$, the ratio of $4$-blocks $\theta\geq\frac{1}{3}$.
\end{lemma}
\begin{proof} \label{smallMPDFpf}
From \cite{AB07}, it is known that there exist $(6h+1,4,1)$-PDF for
$h\equiv0\pmod{2}, h\in[2,42]\setminus \{4,6\}$. Then we only need to construct the remaining cases.
We only list the blocks of the desired $(g,\{3,4\},1)$-{\rm PDF}s for $h\in \{3,4,5\}$ below,
for other values of $h$, the desired blocks are listed in Appendix.
~\\
$h=3~\\
\ \{0,1,8\},\{0,3,5,9\}$.
~\\
$h=4,~\\
\{0,5,11\},\{0,4,12\},\{0,1,3,10\}$.
~\\
$h=5~\\
\{0,4,15\},\{0,1,6,14\},\{0,2,9,12\}$.
\end{proof}

\begin{theorem}\label{all M-PDF}
There exists a $(6h+1,\{3,4\},1)$-{\rm PDF} for each $h\geq 2$ such that the ratio of the $4$-blocks $\theta\geq\frac{1}{14}$.
\end{theorem}
\begin{proof}
The result comes from Theorem \ref{all pdf new} and Lemma \ref{smallMPDF}.
\end{proof}

\section{Application to Perfect  NCW-SOOCs}

Strict optical orthogonal codes (SOOCs for short) were introduced by Zhang  \cite{Zhang99,Zhang99-2}
 for fiber-optic code-division multiple-access (FO-CDMA) networks. Such codes can strictly guarantee both auto-correlation and cross-correlation functions constrained to
have the value one in fully asynchronous data communications and ultra fast switching. The interested reader is referred to \cite{Chu2003,Zhang99,Zhang99-2,Zhang2006}
for the details. Most existing works on SOOC'S have assumed that all code-words have the same weight.
In general, the code size of SOOCs depends on the weights of codewords,
and nonconstant weight SOOCs (NCW-SOOCs for short) can also flexibly support multimedia applications in an OCDM system to meet varieties of transmission performance and traffic
 requirements  \cite{Zhang2006}.

In this section, new perfect NCW-SOOCs will be constructed by using PDFs.

Let  $C_{i}=(u_{1}^{i},u_{2}^{i},...,u_{n}^{i})$
 be a  $(0,1)$ sequence of weight $w_{i}$. For convenience, we use set notation for $C_{i}$, i.e.,
 $C_{i}=(c_{1}^{i},c_{2}^{i},...,c_{w_{i}}^{i})$,
 where $c_{l}^{i}$ denotes the position of the $l$th $``1"$ within $C_{i}$, with $1\leq l \leq w_{i}$.
The following notion is defined in \cite{Zhang99-2,Zhang2006}. For  a  $(0,1)$ sequence
 $C_{i}=(c_{1}^{i},c_{2}^{i},...,c_{w_{i}}^{i})$, $1\leq j<k\leq w_i$, define
$d_{jk}^{i}=v_{k}^{i}-v_{j}^{i}-1$.

\begin{definition} \label{max distance}{\rm (\cite{Zhang99-2})}
Let $\mathcal{C}=\{C_{i}|1\leq i\leq |\mathcal{C}|\}$ be a family of $(0,1)$ code with length $n$
and weight set $W=\{|C_i||1\leq i\leq |C|\}$. The \emph{maximum decoding slot distance} $D$ of code $\mathcal{C}$ is defined as ~

$D={\rm max}\{D^{i}| 1\leq i\leq |\mathcal{C}|\}$, where
 $\ D^{i}={\rm max}\{d_{jk}^{i}|1\leq j\leq k-1$, $2\leq k\leq |C_{i}|\}$.
\end{definition}

 To meet the  multiple quality of services (QoS) requirements, Yang introduced multimedia optical CDMA communication system employing
variable-weight OOCs (VW-OOCs) in \cite{Yang}. The term of variable-weight in \cite{Yang} is the same as nonconstant weight in  \cite{Zhang2006}.
The weight distribution sequence is not included in the definition of  nonconstant weight strict optical orthogonal codes (NCW-SOOCs) in \cite{Zhang2006}.
The following definition of NCW-SOOC is based on VW-OOCs in \cite{Yang} and NCW-SOOCs in \cite{Zhang2006}.
Let  $W=\{w_1,...,w_r\}$ be an ordering of a set of $r$ distinct
integers greater than 1, without loss of generality, we may assume $w_1<w_2<\ldots<w_r$.
Let $Q=(q_1,...,q_r)$ be an $r$-tuple ({\it weight distribution sequence})
of positive rational numbers whose sum is 1.

\begin{definition} \label{ncw-sooc-d}
Given a $(0,1)$ code $\mathcal{C}=\{C_{i}|1\leq i\leq |\mathcal{C}|\}$ of length $n$ and weight set $W$,
 $\mathcal{C}$ is an $(n,W,1,Q)$-{\rm SOOC} if it satisfies~\\
{\rm (1)} {\it weight distribution property}: the ratio of codewords of $\mathcal{C}$ with weight
$w_i$ is $q_i$;\\
{\rm (2)}  $d_{jk}^{i}\neq d_{lm}^{i}$ for  $(j,k)\neq(l,m)$;~\\
{\rm (3)}  $d_{jk}^{i}\neq d_{lm}^{i'}$ for  $1\leq i\neq i' \leq |\mathcal{C}|$;~\\
{\rm (4)}  $n\geq2D+3$.
\end{definition}

\begin{lemma}\label{2D+3}{\rm (\cite{Zhang2006})}
An $(n,W,1,Q)$-{\rm OOC} is an {\rm NCW-SOOC} if and only if $n\geq2D+3$ where $D$ is its maximum decoding slot distance.
\end{lemma}

Let $N_{m}$ be the minimum code length of an  $(n,W,1,Q)$-NCW-SOOC for given values of $M$ and $W$. Then
the NCW-SOOC of size $M$ is called a \emph{perfect}  if $n=N_{m}$.

In \cite{Chu2003}, an equivalence between strict optical orthogonal codes and difference triangle sets
is established. To construct NCW-SOOCs, difference triangle sets with variable sizes will be needed.
The following definition of difference triangle sets with variable sizes is a natural generalization of
difference triangle sets in \cite{Chu2003}.

\begin{definition} \label{G-DTSs}
Let $I$, $J_i, \ 1\leq i\leq I$ be positive integers.
An $(I,\{J_{1},J_{2},...,J_{I}\})$-{\rm DTS} is a set
$\T=\{T_{1},T_{2},...,T_{I}\}$,
where
$T_{i}=\{a_{il}|0\leq l\leq J_{i}\}$, for $1\leq i\leq I$
are sets of integers such that all the differences $a_{il}-a_{il'}$, with $1\leq i\leq I$ and $0\leq l\neq l'\leq J_{i}$ are distinct.
\end{definition}

Ordering the elements of $\Delta_{i}$ and subtracting the smallest from each of them, one can get a DTS in normalized form
$0=a_{i0}<a_{i1}<...<a_{i,J_{i}}$ for all $i$. Let $\T=\{T_{1},T_{2},...,T_{I}\}$ be an
 $(I,\{J_{1},J_{2},...,J_{I}\})$ difference triangle set in its normalized form. Define\\
  $m(\T)={\rm max}\{a_{i,J_{i}}|1\leq i\leq I\}$, \  $M(I,\{J_{1},J_{2},...,J_{I}\})={\rm min}\{m(\T)|\T \ is\ an\  (I,(\{J_{1},J_{2},...,J_{I}\})$-${\rm DTS}\}$.
If $m(\T)=M(I,\{J_{1},J_{2},...,J_{I}\})$, then $\T$ is called \emph{optimal}.

Similar to Theorem 5 and Corollary 1 in \cite{Chu2003}, the following result is obtained.

\begin{theorem}\label{equivalence}
Let $\T$ be an $(I,\{J_{1},J_{2},...,J_{I}\})$-{\rm DTS} with $m(\T)$ defined as above.
An optimal $\T$ is equivalent to a perfect $(2m(\T)+1,\{J_{1}+1,J_{2}+1,...,J_{I}+1\},1)$-{\rm NCW-SOOC}.
\end{theorem}

Let $(g,W,1,Q)$-{\rm PDF} be a $(g,W,1)$-{\rm PDF} with block size  distribution sequence $Q$.
The following result is obtained.

\begin{theorem}\label{PDF to NCWSOOC}
If there exists a  $(g,W,1,Q)$-{\rm PDF}, then there exists a perfect  $(g,W,1,Q)$-{\rm NCW-SOOC}.
\end{theorem}
\begin{proof} It is stated in \cite{Sh07} that a $(g,k,1)$-PDF is an optimal $(\frac {g-1}{k(k-1)},k-1)$-DTS with $m(\T)=\frac {g-1}{2}$.
Similarly, a $(g,W,1,Q)$-PDF with $t$ blocks is an optimal $(t,\{w-1|w\in W\})$-DTS with block size distribution sequence $Q$, and $m(\T)=\frac {g-1}{2}$.
Then, the conclusion comes from Theorem \ref{equivalence}.
\end{proof}

From Theorems \ref{all M-PDF}, \ref{PDF to NCWSOOC}, one can get the following result.

\begin{corollary} \label{NCWSOOC-result}
There exists a perfect optimal $(g,\{3,4\},1,(1-\theta,\theta))$-{\rm NCW-SOOC} with $\theta\geq\frac{1}{14}$ for each $g\equiv1\pmod6, g\geq13$.
\end{corollary}

\section{Conclusions}

In this paper, the constructions of perfect difference matrices and perfect difference families are presented.
A PDM$(3,m)$ exists  for any odd $5\leq m<1000$ with two definite exceptions of $m=9,11$
and $33$ possible exceptions, this greatly improved the known results on PDM$(3,m)$s.
New infinite class of perfect difference families with block size set $K=\{3,4\}$ are obtained. As an application,
perfect $(g,\{3,4\},1,(1-\theta,\theta))$-NCW-SOOCs
with $\theta\geq\frac{1}{14}$ are obtained for each $g\equiv1\pmod6$, $g\geq13$.

\noindent {\bf Acknowledgments} \
Research of Xianwei Sun  was supported  by NSFC (No. 61771354);
Research of Huangsheng Yu was supported   by NSFC (No. 11801103);
 Research of Dianhua Wu  was supported  by  NSFC (No. 12161010).

~\\

\noindent {\bf Appendix: Blocks of the PDFs in Lemma \ref{smallMPDF}}
~\\
$h=6,~\\
\{0,5,16\},\{0,6,18\},\{0,1,4,14\},\{0,2,9,17\}$.
~\\
$h=7~\\
\{0,8,20\},\{0,1,10,15\},\{0,2,18,21\},\{0,4,11,17\}$.
~\\
$h=9~\\
\{0,8,24\},\{0,1,4,22\},\{0,2,11,25\},\{0,6,19,26\},\{0,10,15,27\}$.
~\\
$h=11~\\
\{0,8,24\},\{0,1,3,29\},\{0,5,23,30\},\{0,6,20,33\},\{0,9,19,31\}, \{0,11,15,32\}$.
~\\
$h=13~\\
\{0,11,33\},\{0,29,32,39\},\{0,5,19,31\},\{0,6,30,34\},
\{0,1,21,37\},\{0,8,17,35\},\{0,13,15,38\}$.
~\\
$h=15~\\
\{0,13,39\},\{0,8,35,38\},\{0,4,18,41\},\{0,17,33,45\},
\{0,10,42,44\},\{0,1,6,25\},\{0,7,22,43\},~\\ \{0,9,29,40\}$.
~\\
$h=17~\\
\{0,13,39\},\{0,14,37,47\},\{0,12,15,44\},\{0,11,42,49\},
\{0,28,48,50\},\{0,6,24,51\},\{0,4,21,40\},~\\ \{0,5,30,46\},
\{0,8,9,43\}$.
~\\
$h=19~\\
\{0,19,56\},\{0,8,32,48\},\{0,9,36,54\},\{0,42,47,57\},
\{0,23,49,51\},\{0,22,35,52\},\{0,3,34,41\},~\\ \{0,4,29,50\},
\{0,11,12,55\},\{0,14,20,53\}$.
~\\
$h=21~\\
\{0,19,62\},\{0,7,28,42\},\{0,9,36,54\},\{0,10,40,60\},
\{0,11,34,63\},\{0,15,41,47\},\{0,13,44,46\},~\\ \{0,1,5,58\},
\{0,3,51,59\},\{0,12,37,61\},\{0,16,38,55\}$.
~\\
$h=23~\\
\{0,19,68\},\{0,7,28,42\},\{0,9,36,54\},\{0,10,40,60\},
\{0,11,44,66\},\{0,5,39,62\},\{0,16,64,67\},~\\ \{0,8,37,61\},
\{0,25,31,63\},\{0,1,47,59\},\{0,2,17,43\},\{0,4,56,69\}.$
~\\
$h=25~\\
\{0,23,72\},\{0,9,36,54\},\{0,10,40,60\},\{0,11,44,66\},
\{0,52,57,64\},~\{0,15,58,74\},\{0,13,41,75\},~\\ \{0,37,68,69\},
\{0,42,46,71\},\{0,17,38,73\},\{0,39,63,65\},\{0,3,51,70\},
\{0,8,14,61\}$.
~\\
$h=27$~\\
$\{0,1,81\},\{0,36,55,70\},\{0,14,22,73\},\{0,17,67,79\},\{0,41,66,76\},
\{0,27,56,60\}, \{0,16,68,77\},~\\ \{0,11,58,64\},
\{0,28,49,54\}, \{0,23,63,65\},\{0,3,48,72\},\{0,37,57,75\},
\{0,7,46,78\},\{0,30,43,74\}$.
~\\
$h=29$~\\
$\{0,29,86\},\{0,11,44,66\},\{0,12,48,72\},\{0,13,52,78\},
\{0,14,56,84\},\{0,4,63,73\},\{0,3,38,85\},~\\ \{0,2,21,79\},
\{0,49,64,80\},\{0,1,41,75\},\{0,20,81,87\},\{0,27,32,50\},
\{0,9,54,71\},\{0,8,51,76\},~\\ \{0,30,37,83\}$.
~\\
$h=31$~\\
$\{0,1,93\},\{0,33,55,85\},\{0,21,57,89\},\{0,45,61,71\},
\{0,7,51,90\},\{0,3,63,80\},\{0,14,67,86\},~\\ \{0,2,58,78\},
\{0,38,47,75\},\{0,64,69,82\}, \{0,4,66,74\},\{0,48,79,91\},
\{0,6,46,87\},\{0,34,59,88\},~\\ \{0,11,35,84\},\{0,15,42,65\}$.
~\\
$h=33$~\\
$\{0,1,99\},\{0,46,68,81\},\{0,14,90,93\},\{0,11,75,94\},
\{0,21,47,88\},\{0,7,69,96\},\{0,36,74,78\},~\\ \{0,20,63,65\},
\{0,31,48,97\},\{0,33,73,91\},\{0,29,59,82\},\{0,56,71,95\},
\{0,51,57,85\},~\\ \{0,16,60,70\},\{0,12,84,92\},\{0,32,37,87\},
\{0,9,61,86\}$.
~\\
$h=35$~\\
$\{0,1,105\},\{0,12,87,91\},\{0,15,57,64\},\{0,26,61,100\},
\{0,9,59,92\},\{0,54,68,98\},\{0,2,84,95\},~\\ \{0,16,34,85\},
\{0,6,76,103\},\{0,41,81,101\}, \{0,8,71,96\},\{0,10,62,90\},
\{0,43,46,99\},~\\ \{0,45,58,77\},\{0,22,89,94\},\{0,31,55,78\},
\{0,36,73,102\},\{0,21,38,86\}$.
~\\
$h=39$~\\
$\{0,1,117\},\{0,3,66,101\},\{0,44,86,93\},\{0,20,57,103\},
\{0,15,99,115\},\{0,9,70,111\},~\\ \{0,34,64,114\},\{0,5,22,96\},
\{0,12,72,104\},\{0,27,51,105\},\{0,52,97,110\}\{0,23,59,112\},~\\
\{0,18,28,113\},\{0,2,6,75\},\{0,38,81,106\},\{0,39,65,94\},
\{0,62,76,109\},\{0,40,88,107\},~\\ \{0,8,79,90\},\{0,21,77,108\}$.
~\\

\end{document}